\documentclass[onecolumn,12pt]{IEEEtran} 
\usepackage{color} 
\usepackage{amsfonts,color,pslatex}
\usepackage{amssymb,amsmath,latexsym}

\newcommand{\tr}{{\rm Tr}}
\newcommand{\gf}{{\rm GF}}
\newcommand{\Z}{\mathbb {Z}}

\newcommand{\C}{{\mathcal C}}

\newtheorem{theorem}{Theorem}[section]
\newtheorem{lemma}[theorem]{Lemma}

\newtheorem{example}[theorem]{Example}

\newtheorem{conjecture}[theorem]{Open Problem}

\begin{document}

\title{Optimal Ternary Cyclic Codes from Monomials\thanks{C. Ding's research was supported by 
The Hong Kong Research Grants Council, Proj. No. 600812. T. Helleseth's research was supported by 
the
Norwegian
Research
Council. }}
\author{Cunsheng Ding\thanks{C. Ding is with the Department of Computer Science 
                                                  and Engineering, The Hong Kong University of Science and Technology, 
                                                  Clear Water Bay, Kowloon, Hong Kong, China. Email: cding@ust.hk} and 
             Tor Helleseth\thanks{T. Helleseth is with the Department of Informatics, The University of Bergen, 
                                                  N-5020 Bergen, Norway. Email: Tor.Helleseth@ii.uib.no}}

\date{\today}
\maketitle

\begin{abstract} 
Cyclic codes are a subclass of linear codes and have applications in consumer electronics, 
data storage systems, and communication systems as they have efficient encoding and 
decoding algorithms. Perfect nonlinear monomials were employed to construct optimal 
ternary cyclic codes with parameters  $[3^m-1, 3^m-1-2m, 4]$ by Carlet, Ding and Yuan 
in 2005. In this paper, almost perfect nonlinear monomials, and a number of 
other monomials over $\gf(3^m)$ are used to construct optimal ternary cyclic codes 
with the same parameters.  Nine open problems on such codes are also presented.  
\end{abstract}



\section{Introduction}

Let $q$ be a power of a prime $p$. 
A linear $[n, k, d]$ code over $\gf(p)$ is a $k$-dimensional subspace of $\gf(p)^n$ 
with minimum (Hamming) nonzero weight $d$. 
A linear $[n,k]$ code $\C$ over the finite field $\gf(p)$ is called {\em cyclic} if 
$(c_0,c_1, \cdots, c_{n-1}) \in \C$ implies $(c_{n-1}, c_0, c_1, \cdots, c_{n-2}) 
\in \C$.  
Let $\gcd(n, p)=1$. By identifying any vector $(c_0,c_1, \cdots, c_{n-1}) \in \gf(p)^n$ 
with  
$$ 
c_0+c_1x+c_2x^2+ \cdots + c_{n-1}x^{n-1} \in \gf(p)[x]/(x^n-1), 
$$
any code $\C$ of length $n$ over $\gf(p)$ corresponds to a subset of $\gf(p)[x]/(x^n-1)$. 
The linear code $\C$ is cyclic if and only if the corresponding subset in $\gf(p)[x]/(x^n-1)$ 
is an ideal of the ring $\gf(p)[x]/(x^n-1)$. 
It is well known that every ideal of $\gf(p)[x]/(x^n-1)$ is principal. Let $\C=(g(x))$ be a 
cyclic code, where $g(x)$ is monic and has the least degree. Then $g(x)$ is called the 
{\em generator polynomial} and 
$h(x)=(x^n-1)/g(x)$ is referred to as the {\em parity-check} polynomial of 
$\C$.

The error correcting capability of cyclic codes may not be as good as some other linear 
codes in general. However, cyclic codes have wide applications in storage and communication 
systems because they have efficient encoding and decoding algorithms 
\cite{Chie,Forn,Pran}. 
For example, Reed–-Solomon codes have found important applications from deep-space 
communication to consumer electronics. They are prominently used in consumer 
electronics such as CDs, DVDs, Blu-ray Discs, in data transmission technologies 
such as DSL \& WiMAX, in broadcast systems such as DVB and ATSC, and in computer 
applications such as RAID 6 systems.

Perfect nonlinear monomials were employed to construct optimal 
ternary cyclic codes with parameters  $[3^m-1, 3^m-1-2m, 4]$ by Carlet, Ding and Yuan 
in \cite{CDY05}. 
In this paper, almost perfect nonlinear (APN) monomials and a number of classes of monomials 
over $\gf(3^m)$ will be used to construct many classes of optimal ternary cyclic codes with the 
same parameters.  Nine open problems on such codes are also presented.   

\section{Preliminaries} 

In this section, we fix some basic notation for this paper and introduce almost perfect 
nonlinear and planar functions, and $q$-cyclotomic cosets that will be employed in 
subsequent sections.     

\subsection{Some notation fixed throughout this paper}\label{sec-notations} 

Throughout this paper, we adopt the following notation unless otherwise stated: 
\begin{itemize} 
\item $p$ is an odd prime, and $p=3$ for Section \ref{sec-tern} and subsequent sections.  
\item $q=p^m$, where $m$ is a positive integer.  
\item $n=p^m-1$.  
\item $\Z_n=\{0,1,2,\cdots, n-1\}$ associated with the integer addition modulo $n$ and  
           integer multiplication modulo $n$ operations. 
\item $\alpha$ is a generator of $\gf(q)^*:=\gf(q) \setminus \{0\}$. 
\item $m_a(x)$ is the minimal polynomial of $a \in \gf(q)$ over $\gf(p)$.    
\item $\tr(x)$ is the trace function from $\gf(q)$ to $\gf(p)$.        
\item By the Database we mean the collection of the tables of best linear codes known maintained by 
         Markus Grassl at http://www.codetables.de/.     
\end{itemize}

\subsection{The $p$-cyclotomic cosets modulo $n=p^m-1$}\label{sec-cpsets}

The $p$-cyclotomic coset modulo $n$ containing $j$ is defined by 
$$ 
C_j=\{j, pj, p^2j, \cdots, p^{\ell_j-1}j\} \subset \Z_n 
$$
where $\ell_j$ is the smallest positive integer such that $p^{\ell_j}j \equiv j \pmod{n}$, 
and is called the size of $C_j$. It is known that $\ell_j$ divides $m$. The smallest integer 
in $C_j$ is called the {\em coset leader} of $C_j$. Let $\Gamma$ denote the set of all 
coset leaders. By definition, we have 
$$ 
\bigcup_{j \in \Gamma} C_j =\Z_n.  
$$ 

The following lemma is useful in the sequel. 

\begin{lemma}\label{lem-zero}
Let $q=p^m$ and $n=q-1$. For any $1 \le e \le n-1$ with $\gcd(e, q-1)=2$, the length $\ell_e$ of the $p$-cyclotomic 
coset $C_e$ is equal to $m$. 
\end{lemma} 

\begin{proof} 
By definition, $\ell_e$ is the smallest positive integer such that $e(p^{\ell_e}-1) \equiv 0 \pmod{p^m-1}$. Hence 
$\ell_e$ is the smallest positive integer such that $(p^m-1)|e(p^{\ell_e}-1)$. Since $\gcd(e, q-1)=2$, $\ell_e$ is 
the smallest positive integer such that $\frac{p^m-1}{2}|(p^{\ell_e}-1)$. It then follows that $\ell_e=m$. This 
completes the proof. 
\end{proof}

\subsection{Perfect and almost perfect nonlinear functions on $\gf(q)$}\label{sec-APNPN} 

A function $f: \gf(q) \to \gf(q)$ is called {\em almost perfect 
nonlinear (APN)} if 
$$ 
\max_{a \in \gf(q)^*} \max_{b \in \gf(q)} |\{x \in \gf(q): f(x+a)-f(x)=b\}| =2,  
$$ 
and is referred to as 
{\em perfect 
nonlinear or planar} if 
$$ 
\max_{a \in \gf(q)^*} \max_{b \in \gf(q)} |\{x \in \gf(q): f(x+a)-f(x)=b\}| =1.   
$$ 

There is no perfect nonlinear (planar) function on $\gf(2^m)$. Known APN polynomials  over 
$\gf(2^m)$ can be found in \cite{BD,BC,BCL2,BCL1,BCP,Nybe,Gold,Kasa,Dobb99,Dobb992,HX}.  
However, there are both planar and APN functions over $\gf(p^m)$ for odd primes $p$.  
In the sequel, planar and APN monomials over $\gf(3^m)$ will be employed to construct optimal 
ternary cyclic codes. The results of this paper will be a nice demonstration of applications of APN 
functions in engineering.  

\section{The codes defined by pairs of monomials}\label{sec-codesmono} 

Let $p$ be a prime and let $q=p^m$, where $m$ is a positive integer. Let $m_{\alpha^i}(x)$ denote the 
minimal polynomial over $\gf(p)$ of $\alpha^i$, where $\alpha$ is a generator of $\gf(q)^*$. 
In this paper, we consider the cyclic code of length $n=q-1$ over $\gf(p)$ with generator 
polynomial $m_{\alpha}(x) m_{\alpha^e}(x)$, denoted by $\C_{(1,e)}$, where $1 < e <q-1$ and 
$e \not\in C_1$.  The condition that $e \not\in C_1$ is to make sure that  $m_{\alpha}(x)$ and  
$m_{\alpha^e}(x)$ are distinct. The dimension of $\C_{(1,e)}$ is equal to $n-(m+\ell_e)$, 
where $\ell_e=|C_e|$ and $C_e$ is the $p$-cyclotomic coset modulo $n$ containing $e$.  
The cyclic code $\C_{(1,e)}$ is defined by the pair of monomials $x$ and $x^e$ over $\gf(q)$. 

When $p=2$, it was proved in \cite{CCD,CCZ,LintWils} that the binary code $\C_{(1,e)}$ has 
parameters $[2^m-1, 2^m-1-2m, 5]$ if and only if $x^e$ is an APN monomial over $\gf(2^m)$. 
When $p$ is odd and $x^e$ is a planar monomial over $\gf(q)$, the codes $\C_{(1,e)}$ were 
dealt with in \cite{CDY05,YCD06}. When $p>3$ and $\ell_e=m$, the code $\C_{(1,e)}$ 
has minimum distance 2 or 3 which may not be interesting. In this paper, we study the case that 
$p=3$ only and prove that the ternary cyclic code $\C_{(1,e)}$ is optimal for many classes of 
monomials $x^e$ over $\gf(3^m)$.    

\section{A basic theorem about the ternary cyclic codes $\C_{(1,e)}$}\label{sec-tern} 

From now on we consider only the code $\C_{(1,e)}$ for $p=3$. We are mostly interested in the 
case that $e \not\in C_1$ and $\ell_e=|C_e|=m$. In this case, the dimension of the code is 
equal to $n-2m=q-1-2m$. 

The following theorem is the fundamental result of this paper and will be used frequently in 
subsequent sections. 

\begin{theorem}\label{thm-terncmain}
Let $e \not\in C_1$ and $\ell_e=|C_e|=m$. The cyclic code $\C_{(1,e)}$ has parameters 
$[3^m-1, 3^m-1-2m, 4]$ if and only if the following conditions are satisfied: 
\begin{itemize}
\item[C1:] $e$ is even; 
\item[C2:] the equation $(x+1)^e+x^e+1=0$ has the only solution $x=1$ in $\gf(q)^*$; and 
\item[C3:] the equation $(x+1)^e-x^e-1=0$ has the only solution $x=0$ in $\gf(q)$.  
\end{itemize}
\end{theorem}

\begin{proof} 
Clearly, the minimum distance $d$ of the code $\C_{(1,e)}$ cannot be 1. The code $\C_{(1,e)}$ 
has a codeword of Hamming weight 2 if and only if there exist two elements $c_1$ and $c_2$ 
in $\gf(3)^*$ and two distinct elements $x_1$ and $x_2$ in $\gf(3^m)^*$ such that 
\begin{eqnarray}\label{eqn-bbb0}
\left\{ \begin{array}{ll} 
           c_1x_1+c_2x_2=0 \\
           c_1x_1^e+c_2x_2^e=0. 
           \end{array} 
           \right.            
\end{eqnarray}
Note that $x_1 \ne x_2$. We have that $c_1=c_2$. Hence (\ref{eqn-bbb0}) is equivalent to the 
following set of equations: 
\begin{eqnarray*} 
\left\{ \begin{array}{ll} 
           x_1+x_2=0 \\
           x_1^e+x_2^e=0. 
           \end{array} 
           \right.            
\end{eqnarray*}
It then follows that $\C_{(1,e)}$ does not have a codeword of Hamming weight two if and only 
if $e$ is even. 

The code $\C_{(1,e)}$ 
has a codeword of Hamming weight 3 if and only if there exist three elements $c_1$, $c_2$ and
$c_3$ in $\gf(3)^*$ and three distinct elements $x_1$, $x_2$ and $x_3$ in $\gf(3^m)^*$ such 
that 
\begin{eqnarray}\label{eqn-bbb1}
\left\{ \begin{array}{ll} 
           c_1x_1+c_2x_2+c_3x_3=0 \\
           c_1x_1^e+c_2x_2^e+c_3x_3^e=0. 
           \end{array} 
           \right.            
\end{eqnarray}

Due to symmetry it is sufficient to consider the following two cases. 

\subsubsection*{Case A, where $c_1=c_2=c_3=1$} 
In this case, (\ref{eqn-bbb1}) becomes 
\begin{eqnarray}\label{eqn-bbb2}
\left\{ \begin{array}{ll} 
           1+\frac{x_2}{x_1}+\frac{x_3}{x_1}=0 \\
           1+\left(\frac{x_2}{x_1}\right)^e+\left(\frac{x_3}{x_1}\right)^e=0.              
           \end{array} 
           \right.            
\end{eqnarray} 
Putting $y_1=x_2/x_1$ and $y_2=x_3/x_1$. Then $y_i \not\in \{0,1\}$. Hence (\ref{eqn-bbb2}) 
has a desired solution if and only if the equation 
$$ 
(y+1)^e+y^e+1=0
$$
has a solution $y \in \gf(3^m) \setminus \{0,1\}$. 

\subsubsection*{Case B, where $c_1=c_2=1$ and $c_3=-1$} 
In this case, (\ref{eqn-bbb1}) becomes 
\begin{eqnarray}\label{eqn-bbb3}
\left\{ \begin{array}{ll} 
           1+\frac{x_2}{x_1}-\frac{x_3}{x_1}=0 \\
           1+\left(\frac{x_2}{x_1}\right)^e-\left(\frac{x_3}{x_1}\right)^e=0.              
           \end{array} 
           \right.            
\end{eqnarray} 
Putting $y_1=x_2/x_1$ and $y_2=x_3/x_1$. Then $y_i \not\in \{0,1\}$. Hence (\ref{eqn-bbb3}) 
has a desired solution if and only if the equation 
$$ 
(y+1)^e-y^e-1=0
$$
has a solution $y \in \gf(3^m) \setminus \{0,1\}$. 

By the Sphere Packing bound \cite[Theorem 1.12.1]{HPbook}, the minimum distance $d$ is at most $4$. Hence $d=4$ if 
Conditions C1, C2 and C3 are satisfied. This completes the proof of this theorem.  
\end{proof} 

Any ternary linear code with parameters $[3^m-1, 3^m-1-2m, 4]$ is optimal in the sense that the minimum 
distance is maximal for the fixed length $3^m-1$ and the fixed dimension $3^m-1-2m$. In subsequent 
sections, we will present many classes of optimal ternary cyclic codes with these parameters. 

\section{The optimal cyclic codes defined by planar monomials $x^e$ over $\gf(3^m)$} 

We first prove the following lemma, where the first and the last conclusion should be known. 

\begin{lemma}\label{lem-lemPNAPN} 
If $x^e$ is APN or planar over $\gf(3^m)$, then 
\begin{itemize}
\item $e$ must be even; 
\item $\gcd(e, 3^m-1)=2$;  
\item $\ell_e=|C_e|=m$; and 
\item $e \not\in C_1$. 
\end{itemize}
\end{lemma} 

\begin{proof} 
If $e$ is odd, then the equation $(x+1)^e-x^e=1$ will have three solutions $x=0$ and $x=\pm 1$. 
By definition $x^e$ is not planar and APN. This proves the first conclusion. 

Let $x^e$ be planar or APN. Then $e$ must be even. In this case, the equation $(x+1)^e-x^e=0$ 
has already one solution $x=1$. If $\gcd(e, 3^m-1)>2$, then $\gcd(e, 3^m-1) \ge 4$. In this case, the 
equation $y^e=1$ has at least four distinct solutions, $1$, $y_1=-1$, $y_2$ and $y_3$. Define $x_i=1/(y_i-1)$ 
for all $i \in \{1,2,3\}$. Then $x_1$, $x_2$ and $x_3$ are three distinct solutions of the equation 
$(x+1)^e-x^e=0$. This is contrary to the assumption that $x^e$ is planar or APN.  Hence $\gcd(e, 3^m-1)=2$. 

We now prove the third conclusion. Since $\gcd(e, 3^m-1)=2$, $(3^m-1)/2$ must divide $3^{\ell_e}-1$. 
Hence $\ell_e=m$. 

Note that $3^i e -1$ is odd for any $i$ and $3^m-1$ is even. We have $3^i e \not\equiv 1 \pmod{3^m-1}$. 
Hence $e \not\in C_1$. 
\end{proof} 

The following theorem was proved by Carlet, Ding and Yuan in \cite{CDY05}. For completeness, we report it 
here and present a proof for it. 

\begin{theorem}\label{thm-CDY05}
{\rm (Carlet-Ding-Yuan \cite{CDY05})}   If $x^e$ is planar over $\gf(3^m)$, then $\C_{(1,e)}$ is an optimal 
ternary cyclic code with parameters $[3^m-1, 3^m-1-2m, 4]$.  
\end{theorem}

\begin{proof}
Let $x^e$ be planar. By Lemma \ref{lem-lemPNAPN}, $e$ is even. Notice that $x=0$ is already a solution of the 
equation $(x+1)^e-x^e-1=0$. By the definition of planar functions, this equation cannot have other solutions. 
Hence Condition C3 in Theorem \ref{thm-terncmain} is met. 

We now prove that Condition C2 in Theorem \ref{thm-terncmain} holds. Suppose on the contrary that  
$(-x-1)^e+x^e+1=0$ for some $x \in \gf(q) \setminus \{1\}.$ Then 
$$ 
1^e-x^e=x^e-(-x-1)^e,  
$$   
that is 
$$
[x+(1-x)]^e - x^e = [-x-1 +(1-x)]^e-(-x-1)^e. 
$$
Note that $x \ne -x-1$ as $x \ne 1$. This is contrary to the assumption that $x^e$ is planar. 
Therefore  Condition C2 in Theorem \ref{thm-terncmain} is also satisfied. 

By Lemma \ref{lem-lemPNAPN} and Lemma \ref{lem-zero}, the dimension of this code is equal to $q-1-2m$. The desired conclusions 
then follow from Theorem \ref{thm-terncmain}. 
\end{proof}

The following is a list of known planar monomials over $3^m$: 
\begin{itemize} 
\item $x^2$. 
\item $x^{3^h+1}$, where $m/\gcd(m,h)$ is odd (\cite{DO}).  
\item $x^{(3^h+1)/2}$, where $\gcd(m,h)=1$ and $h$ is odd (\cite{CM}). 
\end{itemize} 
Each of these planar monomials gives a class of optimal ternary cyclic codes with parameters 
$[3^m-1, 3^m-1-2m, 4]$. 
More planar functions can be found in \cite{CM,DY06,ZKW,Zha}. 

\begin{example} 
Let $(m, h)=(4, 3)$ and let $e=(3^h+1)/2=14$. Let $\alpha$ be the generator of $\gf(3^m)^*$ with 
$\alpha^4 +2 \alpha^3 +2=0$. Then $\C_{(1,e)}$ is a ternary cyclic code with parameters $[80, 72, 4]$ 
and generator polynomial $x^8 + 2x^5 + x^3 + 2x^2 + 2$. 

The dual of $\C_{(1,e)}$ is a ternary cyclic code with parameters $[80, 8, 48]$ and weight enumerator 
$$ 
1+1320x^{48}+2400x^{51} +80x^{54}+1920x^{57}+840x^{60}. 
$$ 
This dual code has the same parameters as the best known ternary linear code with parameters $[80, 8, 48]$ 
in the Database. 
The upper bound on the minimum distance $d$ of any ternary linear code of length 80 and dimension 8 
is 49. 
\end{example} 

The following theorem is the partial inverse conclusion of Theorem \ref{thm-CDY05} when $e$ is 
of a special form. 

\begin{theorem}\label{thm-inverseCD05}
Let $e=(3^h+1)/2$, where $1 \le h \le m-1$. Then the ternary cyclic code $\C_{(1,e)}$ has parameters 
$[3^m-1, 3^m-1-2m, 4]$ if and only if $h$ is odd and $\gcd(m, h)=1$, i.e., if and only if $x^e$ is planar 
over $\gf(3^m)$.   
\end{theorem}

\begin{proof}
In the proof of Theorem \ref{thm-terncmain} we see that $\C_{(1,e)}$ has no codeword of Hamming weight 
2 if and only if $e$ is even. Hence $h$ must be odd if $d=4$. 

Notice that multiplying the left sides of the two equations in Conditions C2 and C3 leads to 
\begin{eqnarray*}
\lefteqn{\left( (x+1)^e+x^e+1\right)\left( (x+1)^e-x^e-1\right)} \\  
&=& (x+1)^{2e}-(x^e+1)^2 \\
&=& (x+1)^{3^h+1}-(x^{(3^h+1)/2}+1)^2 \\
&=& x\left( x^{(3^h-1)/2} -1\right)^2.
\end{eqnarray*} 
It then follows that the equation 
$$ 
(x+1)^{2e}-(x^e+1)^2=0 
$$
does not have a solution $x \in \gf(q) \setminus \{0, \pm 1\}$ if and only if 
$$ 
\gcd((3^h-1)/2, 3^m-1)=2, 
$$
which is equivalent to that $h$ is odd and $\gcd(h,m)=1$. 

When $h$ is odd and $\gcd(h,m)=1$, we have $\gcd(e, 3^m-1)=2$. It then follows from Lemma 
\ref{lem-zero} that $\ell_m=m$. Since $1 \not\in C_e$, the dimension of this code is equal to 
$3^m-1-2m$. 
\end{proof}

We shall need the following lemma later. 

\begin{lemma}\label{lem-e31} 
Let $e=3^h+1$, where $0 \le h \le m-1$. When $m$ is odd, $\ell_e=m$. When $m$ is even, 
\begin{eqnarray}
\ell_e = \left\{ \begin{array}{ll} 
                        \frac{m}{2} & \mbox{ if } h=\frac{m}{2} \\
                         m & \mbox{ if } h \ne \frac{m}{2}. 
                         \end{array}
                         \right.                          
\end{eqnarray}
\end{lemma} 

\begin{proof}
First of all, we have 
\begin{eqnarray}
\gcd(3^j-1, 3^m-1)=3^{\gcd(j,m)}-1. 
\end{eqnarray}
It is also known that 
\begin{eqnarray}
\lefteqn{\gcd(3^h+1, 3^m-1)} \nonumber \\ 
&=& \left\{  
\begin{array}{ll}
2                         &  \mbox{ if } m/\gcd(m,h) \mbox{ is odd} \\
3^{\gcd(h,m)}+1 &  \mbox{ if } m/\gcd(m,h) \mbox{ is even.}
\end{array}
\right.  
\end{eqnarray}

When $h \ne m/2$, we have 
$$ 
\gcd(3^h+1, 3^m-1) < 3^{m/2}+1. 
$$
Since $1 \le j < m$, we have 
$$ 
\gcd(3^j-1, 3^m-1)=3^{\gcd(j,m)}-1 \le 3^{m/2}-1. 
$$
It then follows that 
$$ 
\gcd(3^h+1, 3^m-1) \gcd(3^j-1, 3^m-1) < 3^m-1
$$
for all $1 \le j <m$. Hence 
$$ 
\gcd(3^h+1, 3^m-1) \gcd(3^j-1, 3^m-1) \not\equiv 0 \pmod{3^m-1} 
$$
for all $1 \le j <m$. Therefore, $\ell_e=m$ if $h \ne m/2$. 

When $h=m/2$, let $j=m/2$. Then $(3^j-1)(3^h+1) \equiv 0 \pmod{3^m-1}.$ 
It then follows from the discussions above that 
$\ell_e=m/2.$

\end{proof}

\begin{theorem}\label{thm-inverseCD052}
Let $e=3^h+1$, where $0 \le h \le m-1$. Then the ternary cyclic code $\C_{(1,e)}$ has parameters 
$[3^m-1, 3^m-1-2m, 4]$ if and only if $m/\gcd(m, h)$ is odd, i.e., if and only if $x^e$ is planar 
over $\gf(3^m)$.   
\end{theorem}

\begin{proof}
We first prove the necessity of the condition. 
Notice that 
$$ 
\gcd\left( x^{3^h-1}+1,   x^{3^h-1}-1\right)=1. 
$$
We have 
\begin{eqnarray*}
\lefteqn{\gcd\left( x^{3^h-1}+1,   x^{3^m-1}-1\right)} \\ 
&=& \frac{\gcd\left(x^{2(3^h-1)}-1,   x^{3^m-1}-1 \right)}{\gcd\left(x^{3^h-1}-1,   x^{3^m-1}-1 \right)} \\
&=& \frac{x^{\gcd(2(3^h-1), 3^m-1)}-1}{x^{\gcd(3^h-1, 3^m-1)}-1} \\
&=& \frac{x^{\gcd(3^h-1, 3^m-1) \gcd\left(2, \frac{3^m-1}{\gcd(3^h-1, 3^m-1)}\right)}-1} {x^{3^{^{\gcd(h,m)}}-1} -1} \\
&=& \left\{  
\begin{array}{ll}
1                                   & \mbox{ if } m/\gcd(m,h) \mbox{ is odd} \\
x^{3^{^{\gcd(h,m)}}-1}+1  & \mbox{ if } m/\gcd(m,h) \mbox{ is even.}
\end{array}
\right.  
\end{eqnarray*}

Suppose first that $m/\gcd(m,h)$ is even. Then $3^{\gcd(h,m)}-1$ divides $(3^m-1)/2.$ Let $\alpha$ be a generator 
of $\gf(3^m)^*$. Define then $x=\alpha^{\frac{3^m-1}{2(3^{^{\gcd(h,m)}}-1)}}$. Then  
$x^{3^{^{\gcd(h,m)}}-1}+1=0$. 

It is easily checked that 
$$ 
(x+1)^e-x^e-1=x \left( x^{3^h-1} +1 \right).  
$$
Hence $x=\alpha^{\frac{3^m-1}{2(3^{^{\gcd(h,m)}}-1)}}$ is a solution of $(x+1)^e-x^e-1=0$. 
This means that Condition C3 in Theorem \ref{thm-terncmain} is not met. So we have reached 
a contradiction. This 
proves the necessity of the condition in this theorem. 

We now prove the sufficiency of this condition. When $m/\gcd(m,h)$ is odd, it was proved in 
\cite{CM} that $x^e$ is planar. It then follows from Theorem \ref{thm-CDY05} that $\C_{(1,e)}$ 
has parameters $[3^m-1, 3^m-1-2m, 4]$. This completes the proof of this theorem. 
\end{proof}

\section{The optimal cyclic codes defined by APN monomials $x^e$ over $\gf(3^m)$} 

In this section, we present seven classes of optimal cyclic codes defined by APN monomials 
$x^e$ over $\gf(3^m)$. 

\begin{theorem}\label{thm-apntcodes}
The ternary cyclic code $\C_{(1,e)}$ has parameters $[3^m-1, 3^m-1-2m, 4]$ if $x^e$ is APN 
over $\gf(3^m)$.   
\end{theorem}

\begin{proof}
Let $x^e$ be APN over $\gf(q)$, where $q=3^m$. It then follows from Lemma \ref{lem-lemPNAPN} 
that $\gcd(e, 3^m-1)=2$.  By Lemma \ref{lem-zero}, the dimension of this code is equal to $3^m-1-2m$.  

We now prove that Condition C2 of Theorem \ref{thm-terncmain} is satisfied. To this end, we need 
to prove that there does not exist two distinct elements $y$ and $z$ in $\gf(q) \setminus \{0,1\}$ such that 
\begin{eqnarray}\label{eqn-apntc1}
\left\{ \begin{array}{l} 
           1+y+z=0 \\
           1+y^e+z^e=0. 
           \end{array}
\right.            
\end{eqnarray}

Suppose on the contrary that (\ref{eqn-apntc1}) has such a solution $(y, z)$. Then adding $1$ to 
both sides of both equations in (\ref{eqn-apntc1}) yields 
\begin{eqnarray}\label{eqn-apntc2}
\left\{ \begin{array}{l} 
           y-1=1-z \\
           y^e-1=1-z^e. 
           \end{array}
\right.            
\end{eqnarray} 

Define $a=1-y \ne 0$ and $b=1-y^e$. It then follows from (\ref{eqn-apntc2}) that 
\begin{eqnarray}\label{eqn-apntc3}
\left\{ \begin{array}{l} 
           (1-a)^e-1^e=-b \\
           (z-a)^e-z^e=-b. 
           \end{array}
\right.            
\end{eqnarray} 

Adding $z$ and $z^e$ to the first and second equation of (\ref{eqn-apntc1}) gives 
\begin{eqnarray}\label{eqn-apntc4}
\left\{ \begin{array}{l} 
           1-z=z-y=-a \\
           1^e-z^e=z^e-y^e=-b. 
           \end{array}
\right.            
\end{eqnarray} 
It then follows that 
\begin{eqnarray}\label{eqn-apntc5}
(y-a)^e-y^e=-b. 
\end{eqnarray} 

Hence the equation $(x-a)^e-x^e=-b$ has three distinct solutions $1, y$ and $z$. 
This is contrary to the assumption that $x^e$ is APN. Note that $(x-a)^e-x^e=-b$ 
has at most two solutions $x \in \gf(q)$ for any $(a, b) \in \gf(q)^* \times \gf(q)$ 
according to the definition of APN functions. This completes the proof of 
Condition C2 in Theorem \ref{thm-terncmain}. 

Finally we prove that Condition C3 is also satisfied. Suppose on the contrary that 
the equation $(x+1)^e-x^e-1=0$ has a solution $x \in \gf(q)^*$. Then $x \ne \pm 1$. 
Whence $x^2 \ne 1$ and thus $x \ne x^{-1}$. Dividing both sides of the equation 
$(x+1)^e-x^e-1=0$ with $x^e$ yields $(x^{-1}+1)^e-x^{-e}-1=0$. Therefore $x^{-1}$ 
is also a solution 
of the equation $(x+1)^e-x^e-1=0$. So this equation has three distinct solutions 
$0, x$ and $x^{-1}$.  This is contrary to the assumption that $x^e$ is APN over $\gf(q)$. 
Hence Condition C3 in Theorem \ref{thm-terncmain} is indeed met. 

The desired conclusions of this theorem then follow from Theorem \ref{thm-terncmain}. 
\end{proof}

The following is a summary of known APN monomials $x^e$ over $\gf(3^m)$. Each of them gives a 
class of optimal ternary cyclic codes with parameters $[3^m-1, 3^m-1-2m, 4]$.   
\begin{itemize} 
\item $e=3^{m-1}-1$, $m$ is odd. 
\item $e=3^{(m+1)/2}-1$, $m$ is odd. 
\item $e=\frac{3^m-3}{2}$, $m \ge 5$ and $m$ is odd (\cite{HRS}).  
\item $e=\frac{3^m+1}{4} + \frac{3^m-1}{2}$, $m \ge 3$ and $m$ is odd (\cite{HRS}).  
\item $e=\left(3^{(m+1)/4}-1\right)\left(3^{(m+1)/2}+1\right)$, $m \equiv 3 \pmod{4}$ \cite{Zha}. 
\item The exponent $e$ is defined by 
 \begin{eqnarray*} 
e=\left\{ \begin{array}{ll} 
               \frac{3^{(m+1)/2}-1}{2} & \mbox{if } m \equiv 3 \pmod{4} \\
               \frac{3^{(m+1)/2}-1}{2} + \frac{3^m-1}{2} & \mbox{if } m \equiv 1 \pmod{4}.                 
               \end{array} 
\right.                
\end{eqnarray*}   
\item The exponent $e$ is defined by  
\begin{eqnarray*} 
e=\left\{ \begin{array}{ll} 
               \frac{3^{m+1}-1}{8} & \mbox{if } m \equiv 3 \pmod{4} \\
               \frac{3^{m+1}-1}{8} + \frac{3^m-1}{2} & \mbox{if } m \equiv 1 \pmod{4}.                 
               \end{array} 
\right.                
\end{eqnarray*}   
\end{itemize}   

There are other APN monomials over $\gf(p^m)$ for $p \ge 5$. The reader is referred to 
\cite{HRS, Zha} for details. 

\begin{example} 
Let $m=3$ and let $e=3^{(m+1)/2}-1=8$. Let $\alpha$ be the generator of $\gf(3^m)^*$ with 
$\alpha^3 +2\alpha +1=0$. Then $\C_{(1,e)}$ is a ternary cyclic code with parameters $[26, 20, 4]$ 
and generator polynomial $x^6 + 2x^5 + x^4 + x^3 + 2$. 

The dual of $\C_{(1,e)}$ is a ternary cyclic code with parameters $[26, 6, 15]$ and weight enumerator 
$$ 
1+312x^{15}+260x^{18}+156x^{21}. 
$$ 
This code is also optimal, while the optimal ternary cyclic code with the same parameters in the Database 
is not known to be cyclic. 
\end{example}

\section{The optimal cyclic codes defined by other monomials $x^e$ over $\gf(3^m)$} 

In the previous sections, optimal ternary cyclic codes from planar and APN monomials over 
$\gf(3^m)$ were presented. In this section, we construct several classes of  optimal 
ternary cyclic codes with parameters $[3^m-1, 3^m-1-2m, 4]$ using  monomials 
over $\gf(3^m)$ that are neither planar nor APN. 

We prove first the following theorem. 

\begin{theorem}
Let $e=(3^h-1)/2$, where $2 \le h \le m-1$. The ternary cyclic code $\C_{(1,e)}$ has parameters 
$[3^m-1, 3^m-1-2m, 4]$ if and only if 
\begin{itemize}
\item[(a)] $m$ is odd; 
\item[(b)] $h$ is even; 
\item[(c)] $\gcd(h,m)=1$; and 
\item[(d)] $\gcd(h-1, m)=1$. 
\end{itemize} 
\end{theorem}

\begin{proof}
The code $\C_{(1,e)}$ does not have a codeword of Hamming weight two if and only if $e$ is even. 
Clearly $e$ is even if and only if $h$ is even. 

Obviously, Conditions C2 and C3 in Theorem \ref{thm-terncmain} are simultaneously satisfied if 
and only if  
the equation $(x+1)^{2e} - (x^e+1)^2=0$ does not have a solution $x \in \gf(q)\setminus 
\{0, \pm 1\}$. 

Let $x \in \gf(q)\setminus 
\{0, \pm 1\}$. 
Note that 
\begin{eqnarray*} 
\lefteqn{(x+1)^{2e} - (x^e+1)^2} \\  
&=& (x+1)^{3^h-1} -\left( x^{3^h-1} -  x^{(3^h-1)/2} +1 \right) \\ 
&=& \frac{(x+1)^{3^h} -(x+1)( x^{3^h-1} -  x^{(3^h-1)/2} +1)}{x+1} \\
&=& -x(x+1)^{-1} \left(x^{(3^h-3)/2}-1  \right) \left( x^{(3^h-1)/2} - 1 \right). 
\end{eqnarray*}
The equation $(x+1)^{2e} - (x^e+1)^2=0$ does not have a solution $x \in \gf(q)\setminus 
\{0, \pm 1\}$ if and only if 
\begin{eqnarray}\label{eqn-e23a}
\gcd\left(\frac{3^{h-1}-1}{2}, 3^m-1\right) \in \{1, 2\} 
\end{eqnarray}
and 
\begin{eqnarray}\label{eqn-e23b}
\gcd\left(\frac{3^{h}-1}{2}, 3^m-1\right) \in \{1, 2\}.  
\end{eqnarray}

It is easily seen that for any $i \ge 1$ 
\begin{eqnarray*}
\gcd\left(\frac{3^i-1}{2}, 3^m-1\right)= \left\{ 
\begin{array}{l}
\left( 3^{^{\gcd(i,m)}} -1 \right)/2  \mbox{ if $i$ is odd}  \\
3^{^{\gcd(i,m)}} -1  \mbox{ if $i$ is even.} 
\end{array} 
\right. 
\end{eqnarray*}

Conditions (a), (b), (\ref{eqn-e23a}) and (\ref{eqn-e23b}) together are equivalent to Conditions (a), (b), (c) and (d) 
together. 

When all the conditions of this theorem are met, we have 
$$ 
\gcd(e, 3^m-1)=2. 
$$
It then follows from Lemma \ref{lem-zero} that $\ell_e=m$. 

Note that $e=(3^h-1)/2=3^{h-1}+3^{h-2}+\cdots + 3 + 1$. Hence $e \not\equiv 3^i \pmod{3^m-1}$ 
for any $1 \le i \le m-1$. Thus $e \not\in C_1$. 

The desired conclusions in this theorem then follow from Theorem \ref{thm-terncmain}.      

\end{proof}

\begin{example} 
Let $(m, h)=(5, 2)$ and let $e=(3^h-1)/2=4$. Let $\alpha$ be the generator of $\gf(3^m)^*$ with 
$\alpha^5 +2 \alpha +1=0$. Then $\C_{(1,e)}$ is a ternary cyclic code with parameters $[242, 232, 4]$ 
and generator polynomial $x^{10} + 2x^9 + 2x^7 + x^5 + 2x^4 + x^3 + x^2 + 2x + 2$. 

The dual of $\C_{(1,e)}$ is a ternary cyclic code with parameters $[242, 10, 153]$ and weight enumerator 
$$ 
1+21780x^{153}+19844x^{162} +17424x^{171}. 
$$ 
This code has the same parameters as the best known ternary linear code with parameters $[242, 10, 153]$ 
in the Database. 
The upper bound on the minimum distance $d$ of any ternary linear code of length 242 and dimension 10  
is 155. 
\end{example} 

\begin{lemma}\label{lem-e321} 
Let $e=2(3^h+1)$, where $0 \le h \le m-1$. When $m$ is odd, $\ell_e=m$. When $m$ is even, 
\begin{eqnarray}
\ell_e = \left\{ \begin{array}{ll} 
                        \frac{m}{2} & \mbox{ if } h=\frac{m}{2} \\
                         m & \mbox{ if } h \ne \frac{m}{2}. 
                         \end{array}
                         \right.                          
\end{eqnarray}
\end{lemma} 

\begin{proof}
The proof is similar to that of Lemma \ref{lem-e31} and is omitted. 
\end{proof} 

\begin{theorem}\label{thm-e213}
Let $e=2(1+3^h)$, where $0 \le h \le m-1$, and let $m$ be even. 
Then the ternary cyclic code $\C_{(1,e)}$ has parameters $[3^m-1, k, 3]$, 
where 
\begin{eqnarray*}
k=\left\{ \begin{array}{ll}
               3^m-1-2m     & \mbox{ if } h \ne m/2 \\
               3^m-1-3m/2 & \mbox{ if } h = m/2. 
\end{array}
\right. 
\end{eqnarray*}
\end{theorem}

\begin{proof}
It is easily seen that $e \not\in C_1$. 
The conclusions on the dimension then follow from Lemma \ref{lem-e321}. Since $e$ is even, 
the minimum distance $d \ge 3$. We now prove that $d=3$. To this end, we prove that 
Condition C2 or C3 in Theorem \ref{thm-terncmain} is not satisfied. 

Since $m$ is even, $(3^m-1)/2$ is even. Define $x=\alpha^{(3^m-1)/4}$, where $\alpha$ 
is a generator of $\gf(3^m)^*$. Then $x^2=\alpha^{(3^m-1)/2}=-1.$  

When $h$ is even, we have 
\begin{eqnarray*}
\lefteqn{(x+1)^{2(1+3^h)}-x^{2(1+3^h)}-1} \\  
&=& (x^2-x+1)^{1+3^h}-(x^2)^{1+3^h}-1 \\ 
&=& (-x)^{1+3^h}-(-1)^{1+3^h}-1 \\ 
&=& 0. 
\end{eqnarray*}
In this case, Condition C3  in Theorem \ref{thm-terncmain} is not met, and hence $d=3$. 

When $h$ is odd, we have 
\begin{eqnarray*}
(x+1)^{2(1+3^h)}+x^{2(1+3^h)}+1 
= (-1)^{(1+3^h)/2} + 2  
= 0. 
\end{eqnarray*}
In this case, Condition C3  in Theorem \ref{thm-terncmain} is not met, and hence $d=3$. 
 
\end{proof}

The code $\C_{(1,e)}$ of Theorem \ref{thm-e213} is almost optimal when $h \ne m/2$, 
and is optimal when $(m,h)=(2,1)$ and $(m,h)=(4,2)$. 

\begin{conjecture}\label{conj-e213}
Let $e=2(1+3^h)$, where $0 \le h \le m-1$, and let $m$ be odd. Does 
the ternary cyclic code $\C_{(1,e)}$ have parameters $[3^m-1, 3^m-1-2m, 4]$?   
\end{conjecture} 

When $3 \le m \le 13$, the answer to this question is positive and confirmed by Magma. 

\begin{theorem}
Let $e=3^h-1$, where $1 \le h \le m-1$. The ternary cyclic code $\C_{(1,e)}$ has parameters 
$[3^m-1, 3^m-1-2m, 4]$ if and only if 
\begin{itemize}
\item[(A)] $\gcd(h,m)=1$; and 
\item[(B)] $\gcd(3^h-2, 3^m-1)=1$. 
\end{itemize} 
\end{theorem}

\begin{proof}
Let Condtions (A) and (B) be met. We first prove that $\ell_e=m$. To this end, we prove 
that there is no $1 \le j \le m-1$ such that $3^m-1$ divides $(3^j-1)(3^h-1)$. Note that 
\begin{eqnarray*}
\lefteqn{\gcd(3^m-1, (3^j-1)(3^h-1))) } \\ 
&=& \frac{\gcd(3^m-1, 3^h-1)\gcd(3^m-1, 3^{^j}-1)}{\gcd(3^{^j}-1, 3^h-1)} \\
&=& 2\frac{3^{^{\gcd(m,j)}}-1}{3^{^{\gcd(j,h)}}-1} \\
&<& 3^m-1. 
\end{eqnarray*}
Hence $\ell_e=m$. Clearly, $e \not\in C_1$.  Thus the dimension of the code is equal to $3^m-1-2m$. 

It is now time to prove that the minimum distance $d=4$. Suppose on the contrary that the 
equation $(x+1)^e+x^e+1=0$ has a solution $x \in \gf(3^m)\setminus \{0, \pm 1\}$. Then 
we have 
$$ 
(x+1)^{3^h-1}+x^{3^h-1}+1=0. 
$$
Multiplying both sides of this equation with $x+1$ yields 
$$ 
x^{3^h-1}-x^{3^h}+x-1=(1-x)(x^{3^h-1}-1)=0. 
$$
Whence $x^{3^h-1}=1$. It then follows from Condition (A) that $x^2=1$. This is contrary to 
the assumption that $x \ne \pm 1$. Therefore Condition C2 in Theorem \ref{thm-terncmain}  
is met. 

Now suppose on the contrary that the equation $(x+1)^e-x^e-1=0$ has a solution 
$x \in \gf(3^m)\setminus \{0, \pm 1\}$. Then 
we have 
$$ 
(x+1)^{3^h-1}-x^{3^h-1}-1=0. 
$$
Multiplying both sides of this equation with $x(x+1)$ gives  
$$ 
x^{3^h}+x^2=0. 
$$
Whence $x^{3^h-2}+1=0$ and $x^{2(3^h-2)}=1$. It then follows from Condition (B) that $x^2=1$. 
This is contrary to the assumption that $x \ne \pm 1$. Therefore Condition C3 in Theorem 
\ref{thm-terncmain} is satisfied. 

Note that $e$ is even. It then follows from Theorem \ref{thm-terncmain} that $d=4$. 
\end{proof}

\begin{example} 
Let $(m, h)=(5, 2)$ and let $e=3^h-1=8$. Let $\alpha$ be the generator of $\gf(3^m)^*$ with 
$\alpha^5 +2 \alpha +1=0$. Then $\C_{(1,e)}$ is a ternary cyclic code with parameters $[242, 232, 4]$ 
and generator polynomial $x^{10} + 2x^9 + x^6 + x^5 + 2x^4 + x^2 + 2$. 

The dual of $\C_{(1,e)}$ is a ternary cyclic code with parameters $[242, 10, 147]$ and weight enumerator 
\begin{eqnarray*} 
1+2420x^{147}+4840x^{150}+2420x^{153}+7260x^{156}+ \\ 
9680x^{159}+  
10164x^{162} +  
9680x^{165}+2420x^{168}+ \\ 
7260x^{171}+2420x^{174}+484x^{186}. 
\end{eqnarray*} 
\end{example} 

Below we list a number of open problems on the ternary cyclic codes $\C_{(1,e)}$.  

\begin{conjecture} 
Let $e=2(3^{m-1}-1)$. Does the ternary cyclic code $\C_{(1,e)}$ have parameters 
$[3^m-1, 3^m-1-2m, 4]$ if $m\ge 5$ and $m$ is prime? 
\end{conjecture} 

When $m \in \{5, 7, 11, 13\}$, the answer to this question is positive and confirmed by Magma.  

\begin{conjecture} 
Let $e=(3^h+5)/2$, where $1 \le h \le m-1$. Is it true that the ternary cyclic code $\C_{(1,e)}$ 
has parameters $[3^m-1, 3^m-1-2m, 4]$ if   
\begin{itemize}
\item $m$ is odd and $m \not\equiv 0 \pmod{3}$; 
\item $h$ is odd; and 
\item $\gcd(h,m)=1$? 
\end{itemize} 
\end{conjecture} 

For all $5 \le m \le 15$, the answer to this question is positive and confirmed by Magma.  

\begin{conjecture} 
Let $e=(3^{h}-5)/2$, where $2 \le h \le m-1$. Let $m$ be odd and $m \not\equiv 0 \pmod{3}$. Is it true that the ternary cyclic code 
$\C_{(1,e)}$ has parameters $[3^m-1, 3^m-1-2m, 4]$ if $h$ is even?  
\end{conjecture} 

For all $3 \le m \le 15$, the answer to this question is positive and confirmed by Magma.  

\begin{conjecture} 
Let $e=(3^h-5)/2$, where $2 \le h \le m-1$.  Let $m$ be even. What are the conditions on $m$ 
and $h$ under which the ternary cyclic code 
$\C_{(1,e)}$ has parameters $[3^m-1, 3^m-1-2m, 4]$?   
\end{conjecture} 

\begin{conjecture} 
Let $e=3^h+5$, where $2 \le h \le m-1$.  Let $m$ be even. Is it true that the ternary cyclic code 
$\C_{(1,e)}$ has parameters $[3^m-1, 3^m-1-2m, 4]$ if one of the following conditions is 
met?   
\begin{itemize}
\item $m \equiv 0 \pmod{4}$, $m \ge 4$ and $h =m/2$. 
\item $m \equiv 2 \pmod{4}$, $m \ge 6$ and $h =(m+2)/2$.    
\end{itemize} 
\end{conjecture} 

For $m \in \{6, 10, 14, 18\}$, the answer to this question is positive and confirmed by Magma.

\begin{conjecture} 
Let $e=3^h+5$, where $0 \le h \le m-1$.  Let $m\ge 5$ and let $m$ be a prime. Is it true that the 
ternary cyclic code $\C_{(1,e)}$ has parameters $[3^m-1, 3^m-1-2m, 4]$ for all $h$ with 
$0 \le h \le m-1$?  
\end{conjecture} 

For $m \in \{3, 5, 7, 11, 13, 17\}$, the answer to this question is positive and confirmed by Magma.   

\begin{conjecture} 
Let $e=3^h+13$, where $3 \le h \le m-1$.  Let $m$ be an odd prime.  What are the conditions on 
$h$ under which the ternary cyclic code $\C_{(1,e)}$ has parameters $[3^m-1, 3^m-1-2m, 4]$?  
\end{conjecture} 

\begin{conjecture} 
Let $e=(3^{m-1}-1)/2+3^h+1$, where $0 \le h \le m-1$.  What are the conditions on $m$ and $h$ 
under which the 
ternary cyclic code $\C_{(1,e)}$ has parameters $[3^m-1, 3^m-1-2m, 4]$? 
\end{conjecture} 

\section{Concluding remarks} 

When $x^e$ is planar or APN over $\gf(3^m)$, the code $\C_{(1,e)}$ has parameters 
$[3^m-1, 3^m-1-2m, 4]$ and is thus optimal. However, as demonstrated in this paper, 
the code may have the same parameters when $x^e$ is neither planar nor APN. It looks 
hard to completely characterize the ternary cyclic codes $\C_{(1,e)}$ with parameters 
$[3^m-1, 3^m-1-2m, 4]$. In this paper, we presented nine open problems on the 
codes $\C_{(1,e)}$. It would be nice if some of these open problems could be settled.  
Finally, we inform the reader that there are other examples of monomials $x^e$ that 
define optimal cyclic codes $\C_{(1,e)}$ according to our Magma experimental data. 

\section*{Acknowledgments}

The authors are grateful to the reviewers for their comments  
that improved the presentation and quality of this paper.

{\small 
\vspace{.2cm} 
\noindent 
{\bf Cunsheng Ding} (M'98--SM'05) was born in 1962 in Shaanxi, China. 
He received the M.Sc.
degree in 1988 from the Northwestern Telecommunications Engineering 
Institute, Xian, China; and the Ph.D. in 1997 from the University of 
Turku, Turku, Finland.  
From 1988 to 1992 he was a Lecturer of Mathematics at Xidian University, 
China. Before joining the Hong Kong University of Science and Technology 
in 2000, where he is currently Professor of Computer Science and Engineering, 
he was Assistant Professor of Computer Science at the National University 
of Singapore. 

His research fields are cryptography and coding theory. 
He has coauthored 
four research monographs, and served as a guest editor or editor for 
ten journals.  
Dr. Ding co-received the State Natural Science Award of China in 1989.

\vspace{.2cm} 
\noindent 
{\bf Tor Helleseth} (M'89--SM'96--F'97) received the Cand. Real. and Dr. Philos.
degrees in mathematics from the University of Bergen, Bergen, Norway, in 1971
and 1979, respectively.

From 1973--1980, he was a Research Assistant at the Department of Mathematics,
University of Bergen. From 1981--1984, he was at the Chief Headquarters
of Defense in Norway. Since 1984, he has been a Professor in the
Department of Informatics, University of Bergen. During the academic years
1977--1978 and 1992--1993, he was on sabbatical leave at the University of
Southern California, Los Angeles, and during 1979--1980, he was a Research
Fellow at the Eindhoven University of Technology, Eindhoven, The Netherlands.
His research interests include coding theory and cryptology. 

Prof. Helleseth served as an Associate Editor for Coding Theory for the IEEE
TRANSATIONS ON INFORMATION THEORY from 1991 to 1993. He was Program
Chairman for Eurocrypt 1993 and for the Information Theory Workshop in
1997 in Longyearbyen, Norway. He was a Program Co-Chairman for SETA04
in Seoul, Korea, and SETA06 in Beijing, China. He was also a Program
Co-Chairman for the IEEE Information Theory Workshop in Solstrand,
Norway in 2007. During 2007--2009 he served on the Board of Governors for
the IEEE Information Theory Society. In 1997 he was elected an IEEE Fellow
for his contributions to coding theory and cryptography. In 2004 he was elected
a member of Det Norske Videnskaps-Akademi. 
}
\end{document}